 \documentclass[a4paper,preprint,10pt]{elsarticle}

\usepackage{epsfig}
\usepackage{amssymb}
\usepackage{amsthm}
\usepackage{subfigure}
\usepackage{amsmath}

\usepackage{amssymb}
\usepackage{amsfonts}
\usepackage{graphicx}
\usepackage{url}
\usepackage{ccaption}
\usepackage{booktabs}
\usepackage{color}
\usepackage{multirow}

\journal{Journal of The Franklin Institute}

\begin{document}

\begin{frontmatter}



\title{
Composite Adaptive Control for Bilateral Teleoperation Systems without Persistency of Excitation\tnoteref{mytitlenote}}


\author[YL]{Yuling Li}\ead{lyl8ustb@gmail.com}
\author[YL]{Yixin Yin}\ead{yyx@ies.ustb.edu.cn}
\author[YL]{Sen Zhang}\ead{zhangsen@ustb.edu.cn}
\author[YL]{Jie Dong\corref{cor1}}\ead{dongjie@ies.ustb.edu.cn}
\author[RJ]{Rolf Johansson}\ead{Rolf.Johansson@control.lth.se}

\cortext[cor1]{Corresponding author}

\address[YL]{School of Automation and Electrical Engineering,
University of Science and Technology Beijing, Beijing, 100083, P.~R.~China}
\address[RJ]{ Department of Automatic Control, Lund University, P.O. Box 118, 22100 Lund, Sweden. }

\newtheorem{assumption}{Assumption}
\newtheorem{definition}[assumption]{Definition}
\newtheorem{theorem}[assumption]{Theorem}
\newtheorem{remark}[assumption]{Remark}
\newtheorem{proposition}[assumption]{Proposition}

\begin{abstract}
Composite adaptive control schemes, which use both the system tracking errors and the prediction error to drive the update laws, have become widespread in achieving an improvement of system performance. However, a strong persistent-excitation (PE) condition should be satisfied to guarantee the parameter convergence. This paper proposes a novel composite adaptive control to guarantee parameter convergence without PE condition for nonlinear teleoperation systems with dynamic uncertainties and time-varying communication delays. The stability criteria of the closed-loop teleoperation system are given in terms of linear matrix inequalities. New tracking performance measures are proposed to evaluate the position tracking between the master and the slave. Simulation studies are given to show the effectiveness of the proposed method.
\end{abstract}


\end{frontmatter}


\section{Introduction}
Bilateral teleoperation systems are  one of the most well-known robotic systems  which extend human operators' intelligence and manipulation skills to the remote environments. A typical single-master-single-slave teleoperation system composes of five parts: a human operator, a master manipulator, communication channel, a slave manipulator and a task environment. The master is directly operated by a human operator to manipulate the slave in the task environment, and the signals (positions, velocities or interaction forces) from the slave are sent back to the master to improve the manipulation performance. Recent years have witnessed considerable advances in the control studies of teleoperation systems, owing to their broad engineering applications in telesurgery, space exploration, nuclear operation, undersea exploration, and so forth.

In practice of teleoperation, there usually exist modeling uncertainties caused by inaccurate parameters of links, unknown load, etc., in the master and slave robots. Besides, the slave robot often interacts with unknown environments, which also leads to uncertainties in the robot manipulators as well.  Adaptive control is an effective technique for handling structured uncertainties and has obtained widespread applications in manipulators, where some recent results can be found in \cite{Islam2011adaptive, li2014decentralised, dehghan2015adaptive, Li2017adaptive}.
 Classical adaptive control usually uses tracking error feedback to update the adaptive estimates.
The use of the tracking error is motivated by the need to cancel cross-terms in the closed-loop tracking error system within a Lyapunov-based analysis. However, these methods can only guarantee the state tracking of the system, while the parameter convergence to their true values is kind of problematic. It is well known that the convergence of the parameters to their true values can improve system performance with accurate online identification, exponential tracking and robust adaption without parameter drift. However, these features are not guaranteed unless a condition of persistent excitation (PE) is satisfied \cite{Pan2016TAC}. It is well known that the PE condition is very stringent and often infeasible in practical control systems \cite{pan2016ccdc}.

In another line, due to the nature of the long-distance data transmission, communication time delays should not be neglected in the control of teleoperation systems \cite{liu2017stability, liu2016generalized, Liu2017comparison}.  The master-slave synchronization \cite{yan2017}  and stability analysis of teleoperation systems with various kinds of time delays, such as constant delays \cite{Spong1989_classical}, time-varying delays\cite{Hua2010_LMI,  li2015guaranteed, yang2015_syn, SARRAS2014_JFI, zhai2015_IJRNC, zhai2016_inputsaturation}, or stochastic delays\cite{kang2013_stodelay, li2013_stodelay, li2016_unknownconstraints}, and so on, have been hot topics in the study of teleoperation systems in recent years.  To handle parameter uncertainties and communication delays in a unified framework, classical adaptive control with sliding mode is introduced into the control of teleoperation systems.
For example, Chopra et. al. \cite{chopra2008brief} proposed an adaptive controller for teleoperation systems with constant time delays and without using scattering transformation. Nu\~no et. al \cite{Nuno2010improvedsyn} pointed out the limitation of Chopra's results in  \cite{chopra2008brief} and proposed a more general adaptive controller for nonlinear teleoperation systems  with constant delays, and further with time-varying delays \cite{SARRAS2014_JFI}. However, these classical adaptive control schemes only achieve asymptotic convergence of master-slave tracking errors, while parameter convergence is seldom considered. With the motivation of using more information to update the parameter estimates to obtain an improved tracking performance, composite adaptive control has been used in teleoperation systems \cite{kim2013apt, chen2016adaptive}.  However, in the works \cite{kim2013apt, chen2016adaptive}, the communication delays were assumed to be constant, which is unrealistic in applications. Moreover, the PE condition or a relaxed sufficient excitation (SE) condition was still required for parameter convergence in these works.

Motivated by the above-mentioned facts, in this paper, a new composite adaptive controller  is designed for teleoperation systems with time-varying delays. This paper has some unique features and key contributions over the existing works in the following ways. First, contrary to the existing works which guarantee the boundedness of the parameter estimation errors \cite{chopra2008brief, Nuno2010improvedsyn, SARRAS2014_JFI}, this paper achieves convergence of parameters to their true values when the communication delays are time-varying, which then gives rise to an improvement of system performance. Second, a new prediction error is designed to guarantee the parameter convergence, and neither PE or SE condition is required, thus the proposed control scheme is more practical in real applications.  Third, the derivatives of the time-varying communication delays are not needed in the controller formulation, which thus makes the controller easier to implement since the information of the delays' derivatives is not easy to obtain in real applications.

The arrangement of  this paper is as follows. In Section \ref{sec:pro_formulation} the system modeling and some preliminaries are given. In Section \ref{sec:controldesign}, the adaptive control with parameter convergence is proposed.  Section \ref{sec:stability} summarizes the stability results of the closed-loop system, while new tracking measures are proposed in Section \ref{sec:trackingindex} to evaluate the delayed tracking performance.
A simple teleoperation system composed of two robots with two degrees of freedom is given in Section \ref{sec:simulations} as an example to show the effectiveness of the proposed method. Finally, the summary and conclusion of this paper is given in Section \ref{sec:conclusion}.

\textbf{Notations}: Throughout this paper, the superscript $T$ stands for matrix transposition, $I$ is used to denote the identity matrix with appropriate dimensions. $*$ represents a block matrix which is readily referred by symmetry. $\mathbb{R}^n$ denotes the $n$-dimensional Euclidean space with the vector norm $\|\cdot\|$, $\mathbb{R}^{n\times m}$ is the set of all $n\times m$ real matrices.  $\lambda_{\min}(M)$ and $ \lambda_{\max}(M)$ denote the maximum and the minimum eigenvalue of matrix $M=M^T\in\mathbb{R}^{n\times n}$, respectively. For any function $f:[0,\infty)\rightarrow \mathbb{R}^n$, the $\mathcal{L}_\infty$-norm is defined as $\|f\|_\infty:=\sup_{t\geq 0}|f(t)|$, and the square of the $\mathcal{L}_2$-norm as $\|f\|_2^2:=\int_{0}^\infty |f(t)|^2dt$. The $\mathcal{L}_\infty$ and $\mathcal{L}_2$ spaces are defined as the sets $\{f:[0,\infty)\rightarrow \mathbb{R}^n,\|f\|_\infty<\infty\}$ and $\{f:[0,\infty)\rightarrow \mathbb{R}^n,\|f\|_2<\infty\}$, respectively.

\section{Problem Formulation and Preliminaries}\label{sec:pro_formulation}

Consider teleoperation systems described as follows:
\begin{align}
&M_m(q_m)\ddot{q}_m\!+\!C_m(q_m,\dot{q}_m)\dot{q}_m\!+\!G_m(q_m)\!=\!F_m\!+\!\tau_m\label{eq:master}\\
&M_{s}(q_{s})\ddot{q}_{s}\!+\!C_{s}(q_{s},\dot{q}_{s})\dot{q}_{s}\!+\!G_{s}(q_{s})=F_{s}\!+\!\tau_{s}\label{eq:slave}
\end{align}
where $q_i,\dot{q}_i,\ddot{q}_i\in\mathbb{R}^n$ are the joint positions, velocities and accelerations of the master and slave devices with $i=m$ or ${s}$ representing the master or the slave robot manipulator respectively. Similarly,  $M_{i}$ represents the mass matrix, $C_{i}(q_i, \dot{q_i})$ embodies the Coriolis and centrifugal effects, $\tau_i$ is the control force, and finally $F_m, F_{s}$ are the external forces applied to the manipulator end-effectors. Each robot in (\ref{eq:master}) and (\ref{eq:slave}) satisfies the structural properties of robotic systems, i.e., the following properties \cite{Nuno2011_tutorial}, \cite{Spong2006_book}:
\begin{enumerate}
\item[P1.] The inertia matrix $M_i(q_i)$ is a symmetric positive-definite function and is lower and upper bounded. i.e., $0<\rho_i^m I\leq M_i(q_i)\leq \rho_i^M I<\infty$, where $\rho_i^m,\rho_i^M$ are positive scalars.
\item[P2.] The matrix $\dot{M}_i(q_i)-2C_i(q_i,\dot{q_i})$ is skew symmetric.
\item[P3.] For all $q_i, x, y\in\mathbb{R}^{n\times 1}$, there exists a positive scalar $c_i$ such that $\|C_i(q_i,x)y\|\leq c_i\|x\|\|y\|$.
\item[P4.] The equations of motion of $n-$link robot can be linearly parameterized as
\begin{align}
M_i(q_i)\ddot{q}_i\!+\!C_i(q_i,\dot{q}_i)\dot{q}_i\!+\!G_i(q_i)\!=\!Y_{io}(q_i,\dot{q}_i,\ddot{q}_i)\theta_i\triangleq y_i,\label{eq:regeressor}
\end{align}
where $Y_{io}(q_i,\dot{q}_i,\ddot{q}_i)\triangleq Y_{io}\in \mathbb{R}^{n\times n_i} $ is a matrix of known functions called regeressor, and $\theta_i\in\mathbb{R}^{n_i}$ is a vector of unknown parameters.
\end{enumerate}

In this paper, we assume that  the master and the slave are coupled with a communication network with time-varying time delays. Hence, the following standard assumption is used.
\begin{assumption}\label{amp:delay}
There exist positive constants $h_i$ and $d_i$ such that the variable communication time-delays $T_i(t)$ satisfies
\begin{eqnarray}
0\leq T_i(t)\leq h_i,\label{eq:BoundofDelay}\\
|\dot{T}_i(t)|\leq d_i<1.\label{eq:BoundofDelaysderivative}
\end{eqnarray}

\end{assumption}

\section{Adaptive Control Design}\label{sec:controldesign}
Suppose the positions  of the master and the slave are available for measurement and are  transmitted through the delayed network communication.
Let $e_i\in\mathbb{R}^n$ denote the position errors by
\begin{align}
e_m\triangleq q_m-q_s(t-T_s(t))\\
e_s\triangleq q_s-q_m(t-T_m(t))
\end{align}
and the velocity errors by
\begin{align}
e_{vm}\triangleq \dot{q}_m-\dot{q}_s(t-T_s(t))\\
e_{vs}\triangleq \dot{q}_s-\dot{q}_m(t-T_m(t))
\end{align}
and then
\begin{align}
\dot{e}_m=\dot{q}_m-(1-\dot{T}_s(t))\dot{q}_s(t-T_s(t))\label{eq:dem}\\
\dot{e}_s=\dot{q}_s-(1-\dot{T}_m(t))\dot{q}_m(t-T_m(t))\label{eq:des}
\end{align}

we define the following auxiliary variables:
\begin{align}
\eta_{m}&\triangleq\dot{q}_m+\lambda_me_{m}\label{eq:eta_m}\\
\eta_{s}&\triangleq\dot{q}_{s}+\lambda_se_{s}\label{eq:eta_s}
\end{align}
where $\lambda_m,\lambda_s$ are positive definite matrices.
By using Property P4, letting
\begin{eqnarray}
Y_i\theta_i&=&Y_i(q_i,\dot{q}_i,e_i,e_{vi})\theta_i\nonumber\\
&=&M_i(q_i)\lambda_ie_{vi}+C_i(q_i,\dot{q}_i)\lambda_ie_i-G_i(q_i),\label{eq:linearparametrization}
\end{eqnarray} for $i=m,s$,
the following control laws for the master and the slaves are proposed:
\begin{align}
\tau_m&=-Y_m\hat{\theta}_m-K_m\eta_{m}\label{eq:taum}\\
\tau_{s}&=-Y_{s}\hat{\theta}_{s}-K_{s}\eta_{s}\label{eq:taus}
\end{align}
where $\hat{\theta}_i$ is the estimate of $\theta_i$, $0<K_i\in\mathbb{R}^{n\times n}$.

Substituting the control law (\ref{eq:taum}-\ref{eq:taus}) into the teleoperation dynamics (\ref{eq:master}-\ref{eq:slave}), we obtain  the following dynamics for $t>0$:
\begin{align}
\left\{
\begin{aligned}
&M_m(q_m)\dot{\eta}_m+C_m(q_m,\dot{q}_m)\eta_m+K_m\eta_m\\&=Y_m\tilde{\theta}_m+F_m+\lambda_mM_m(q_m)\dot{T}_s\dot{q}_s(t-T_s(t))
\\
&M_{s}(q_{s})\dot{\eta}_{s}+C_{s}(q_{s},\dot{q}_{s})\eta_{s}+K_s\eta_{s}\\&=Y_{s}\tilde{\theta}_{s}+F_{s}+\lambda_sM_s(q_s)\dot{T}_m\dot{q}_m(t-T_m(t))\end{aligned}
\right.\label{eq:closeddynamics}
\end{align}
where $\tilde{\theta}_i\triangleq \theta_i-\hat{\theta}_i$.
\begin{remark}
Compared with the existing work \cite{sarras2014adaptive} which considers time-varying communication delays for teleoperation systems, $Y_i$ does not depend on the derivative of position error $e_i$ in the linear-parametrization term (\ref{eq:linearparametrization})  in this paper, which means that the time derivative of the communication delays are not required to formulate the matrix $Y_i$. This makes the controller more suitable for real applications since the values of the time delays  and the derivatives of  the time delays are not obtainable in real applications.
\end{remark}

A straightforward choice of the adaptive law $\dot{\hat{\theta}}_i=\Gamma_iY_i^T\eta_i$ was first proposed by J. J. Slotine \cite{slotine1991applied} and has been widely used in adaptive control of teleoperation systems \cite{Nuno2011_tutorial, Nuno2010improvedsyn, sarras2014adaptive}.
However, it is pointed that this adaptive law cannot guarantee accurate estimations of parameters. In order to achieve convergence of parameters to their true values, the estimation error $\tilde{\theta}_i$ should be introduced into the control design. However, the value of $\tilde{\theta}_i$ is not obtainable since the true value of $\theta_i$ is not available, and thus a prediction error $e_{io}=Y_{io}\tilde{\theta}_i$ or its filtered counterpart $e_{iw}=\alpha\int_{0}^te^{-\delta(t-v)}e_{io}(v)dv$ is used to improve the tracking performance. However, the use of $e_{io}$ or $e_{iw}$ still needs the PE condition to make the system exponential stable.  In the following, we introduce an auxiliary variable  $z_i$ such that $z_i=P_i\tilde{\theta}_i$,where $P_i$ is a designed lower bounded positive-definite matrix, to adaptive control of the teleoperation system. Thus, the following adaptive laws are proposed for the master and the slave,
\begin{align}
\dot{\hat{\theta}}_i&=\Gamma_i(Y_i^T\eta_i+(\xi_i+\delta_i)z_i)\label{eq:updatelaw_m}\\
\dot{z}_i&=-\mu_iz_i+Y_{io}^Te_{io}-P_i\dot{\hat{\theta}}_i,z_i(0)=0\label{eq:updatelaw_z}\\
\dot{P}_i&=-\mu_iP_i+Y_{io}^TY_{io},P_i(0)=P_{i0}>0\label{eq:update_P}\\
\mu_i&=\mu_{i0}(1-\kappa_{i0}\|P_i^{-1}\|)\label{eq:update_lambda}
\end{align}

where
\begin{align}
e_{mo}\triangleq y_m-Y_{mo}\hat{\theta}_m=Y_{mo}\tilde{\theta}_m\label{eq:predictionerror_m}\\
e_{so}\triangleq y_s-Y_{so}\hat{\theta}_s=Y_{so}\tilde{\theta}_s\label{eq:predictionerror_s}
\end{align}
and $\kappa_{i0}$ and $\mu_{i0}$ are two positive constants specifying the lower bound of the norm of $P_i$ and the maximum forgetting rate \cite{slotine1991applied}, $\delta_i$ is a positive constant. $\Gamma_m\in\mathbb{R}^{n_m\times n_m}$ and  $\Gamma_s\in\mathbb{R}^{n_s\times n_s}$ are two constant positive definite matrices.  From Eqs. (\ref{eq:update_P}) and (\ref{eq:update_lambda}), one can show that $\forall t\geq 0$, $\mu_i\geq 0$, and $P_i\geq \kappa_{i0}I$.

  The coefficient $\xi_i$ is given by
  \begin{align}
  \xi_i=\alpha_i\frac{\|Y_i^T\eta_i\|}{\kappa_{i0}}.
  \end{align}
  where $\alpha_i>0$ is a constant.

\begin{remark}
By (\ref{eq:predictionerror_m}-\ref{eq:predictionerror_s}), it is easy to see that the prediction error $e_{io}$ is related to the regressor $Y_i$,
which requires the information of joint acceleration. To avoid this,  the adaptive law (\ref{eq:updatelaw_m}-\ref{eq:update_lambda}) with filtered torques  and filtered regressor ${Y}_{iw}$ could be used. The filtered prediction errors of estimated parameters are defined as
\begin{align}
e_{mw}\triangleq y_{mw}-Y_{mw}\hat{\theta}_m=Y_{mw}\tilde{\theta}_m\\
e_{sw}\triangleq y_{sw}-Y_{sw}\hat{\theta}_s=Y_{sw}\tilde{\theta}_s
\end{align}
where $y_i$ is the filtered forces $\tau_i+F_i$, i.e.,
\[y_{iw}=\alpha\int_0^te^{-\alpha(t-\delta)}y_{i}d\delta\]
and can be calculated without acceleration terms $M_i(q_i)\ddot{q}_i$ by convolving both sides of (\ref{eq:regeressor}) by a filter $W(s)=\frac{\alpha}{s+\alpha}$ \cite{kim2013design}.

\end{remark}

\section{Stability Analysis}\label{sec:stability}

Denote $x_m=[\eta_m^T,\tilde{\theta}_m]^T$,  $x_s=[\eta_s^T,\tilde{\theta}_s]^T$, $x=[x_m^T, x_s^T]^T$, and define the new state $x_t(s)\triangleq x(t+s), s\in[-h,0]$ which take values in $C([-h,0];\mathbb{R}^{n+m+s})$, $h=\max\{h_m, h_s\}$

The following theorem summarizes the stability result of the consider teleoperation system when it is in free motion.
\begin{theorem}\label{Thm:mainresult}
Consider the bilateral teleoperation system (\ref{eq:master}-\ref{eq:slave}) controlled by (\ref{eq:taum}-\ref{eq:taus}) together with the updating law (\ref{eq:updatelaw_m}-\ref{eq:update_lambda}) under the communication channel satisfying Assumption \ref{amp:delay}, if there exist positive-definite matrices $R_{m},R_{s}$ such that the following linear matrix inequality (LMI) holds, respectively:
\begin{eqnarray}
\Pi=\begin{bmatrix}\Pi_1 & 0 &0 & -I\\
* & \Pi_2 &-I & 0\\
* & *& -\frac{R_m}{h_m} & 0\\
* & * & * &-\frac{R_s}{h_s} \end{bmatrix}<0,\label{eq:LMI_PI}
\end{eqnarray}
with
\begin{align*}
\Pi_1=-\frac{1}{\lambda_m}I+h_mR_m+\frac{\lambda_s(\rho_s^M)^2d_m^2}{(1-d_m)k_s^2}I, \\
\Pi_2=-\frac{1}{\lambda_s}I+h_sR_s+\frac{\lambda_m(\rho_m^M)^2d_s^2}{(1-d_s)k_m^2}I,
\end{align*}
then the following claims hold  if the teleoperation system is in free motion, that is, $F_m=F_s=0$:\begin{enumerate}
\item all the signals are bounded and the position errors, velocities and the estimation errors asymptotically converge to zero, that is, $\lim_{t\rightarrow \infty} q_m-q_s=\lim_{t\rightarrow \infty}\dot{q}_m=\lim_{t\rightarrow \infty}\dot{q}_s=\lim_{t\rightarrow \infty}\tilde{\theta}_m=\lim_{t\rightarrow \infty}\tilde{\theta}_s=0$.
\item the estimation errors converge into a specified domain within a given time.    \label{it:pro_item1}

\end{enumerate}
\end{theorem}

\begin{proof} Defining the following function:
\begin{equation}
V_i(x,t)=\frac{1}{2}\eta_i^TM_i(q_i)\eta_i+\frac{1}{2}\tilde{\theta_i}^T\Gamma_i^{-1}\tilde{\theta}_i
\end{equation}
It is obvious that $V_i$ is positive definite and radially unbounded with regard to $\eta_i$ and $\tilde{\theta}_i$.
By using the Property P2,  the derivative of $V_i$ along the trajectory of system (\ref{eq:closeddynamics}) is
 \begin{eqnarray*}
 \dot{V}_i(x,t)&=&-\eta_i^TK_i\eta_i+\eta_i^T(\lambda_iM_i(q_i)\dot{T}_j(t)\dot{q}_j(t-T_j(t))\\
 &&-\tilde{\theta}_i^T(\xi_iP_i+\delta_iP_i)\tilde{\theta}_i\\
 &\leq & -\frac{k_i}{2}\eta_i^T\eta_i-\delta_i\kappa_{i0}\tilde{\theta}_i^T\tilde{\theta}_i\\
 &&+\frac{\lambda_i^2 }{2k_i}(\rho_i^M)^2d_j^2\dot{q}_j^T(t-T_j(t))\dot{q}_j(t-T_j(t))
 \end{eqnarray*}
 where $k_i=\lambda_{\min}\{K_i\}$,
 when $F_m\equiv 0, F_s\equiv 0$.
Now we give the following Lyapunov functional
\begin{equation}
V=V_1+V_2+V_3+V_4\label{eq:V}\end{equation} with
\begin{eqnarray}
V_1&=&\frac{2}{k_m\lambda_m}V_m(x,t)+\frac{2}{k_s\lambda_s}V_s(x,t)\label{eq:V1}\\
V_2&=&(q_m-q_s)^T(q_m-q_s)\label{eq:V2}\\
V_3&=&\sum_{i=m,s}\int^0_{-{h_i}}\int^t_{t+\theta}\dot{q}_i^T(s)R_{i}\dot{q}_{i}(s)dsd\theta\label{eq:V3}\\
V_4&=&\sum_{i=m,s}\int_{t-T_i(t)}^t\nu_i\dot{q}_i^T(s)\dot{q}_i(s)ds\label{eq:V4}
\end{eqnarray}
where  $\nu_i=\frac{\lambda_j(\rho_j^M)^2d_i^2}{(1-d_i)k_j^2}$, $i,j=m,s, i\neq j$. Obviously, $\nu_i>0$.

When the external forces $F_m\equiv F_s\equiv 0$, by (\ref{eq:eta_m}),
the derivative of $V_1$ along with the trajectory of system  (\ref{eq:closeddynamics}) is given by
 \begin{eqnarray*}
\dot{V}_1(x,t)&=&\frac{2}{k_m\lambda_m}\dot{V}_m(x,t)+\frac{2}{k_s\lambda_s}\dot{V}_s(x,t)\\
&\leq &-\sum_{i=m,s}(\frac{\dot{q}_i^T\dot{q}_i}{\lambda_i}+2e_i^T\dot{q}_i+\lambda_ie_i^Te_i\\
&&+\frac{2\delta_i\kappa_{i0}}{k_i\lambda_i}\tilde{\theta}_i^T\tilde{\theta}_i-\frac{\lambda_i}{k_i^2}(\rho_i^M)^2d_j^2\|\dot{q}_j(t-T_j(t))\|^2)
\end{eqnarray*}

It is noted that the position error $e\triangleq q_m-q_s$ can be expressed as
\begin{align}
e=e_m-L_s=-e_s+L_m\label{eq:e}
\end{align}
where $L_m=\int_{t-T_m(t)}^t\dot{q}_m(s)ds$, $L_s=\int_{t-T_s(t)}^t\dot{q}_s(s)ds$, hence the time derivative of $V_2$  along with the trajectory of system  (\ref{eq:closeddynamics}) is given by
\begin{align*}
\dot{V}_2=2(e_m-L_s)^T\dot{q}_m+2(e_s-L_m)^T\dot{q}_s
\end{align*}

Calculating the time derivative of $V_3$, one has that
\begin{eqnarray*}
\dot{V}_3&=&\sum_{i=m,s}h_i\dot{q}_i^TR_{i}\dot{q}_i-\int_{t-h_i}^t\dot{q}_i^T(s)R_{i}\dot{q}_i(s)ds\\
&\leq&\sum_{i=m,s}h_i\dot{q}_i^TR_{i}\dot{q}_i-\int_{t-T_i(t)}^t\dot{q}_i^T(s)R_{i}\dot{q}_i(s)ds
\\&\leq&\sum_{i=m,s}h_i\dot{q}_i^TR_{i}\dot{q}_i-\frac{1}{h_i}L_iR_iL_i
\end{eqnarray*}
by Jensen's inequality.

The derivative of $V_4$ is given by
\[
\dot{V}_4\leq \sum_{i=m,s}\nu_i\dot{q}_i^T\dot{q}_i-\nu_i(1-d_i)\dot{q}_i^T(t-T_i(t))\dot{q}_i(t-T_i(t))
\]

Thus, we have
\begin{eqnarray}
\dot{V}=\sum_{i=1,2,3,4}\dot{V}_i\leq -\xi\Pi\xi-\sum_{i=m,s} (\frac{2\delta_i\kappa_{i0}}{k_i\lambda_i}|\tilde{\theta}_i|^2+\lambda_ie_i^Te_i)\label{eq:dotV}
\end{eqnarray}
where $\xi=\text{col}\{\dot{q}_m,\dot{q}_s,L_m,L_s\}$, and $\Pi$ is given in (\ref{eq:LMI_PI}).

By  (\ref{eq:LMI_PI}), we have that $\dot{V}<0$ and $V>0$. Hence, all the signals are bounded. Furthermore, by (\ref{eq:V}) and  (\ref{eq:dotV}), one has that $e_i\in\mathcal{L}_2\cap\mathcal{L}_\infty$, $\tilde{\theta}_i\in\mathcal{L}_2\cap\mathcal{L}_\infty$, $\dot{q}_i\in\mathcal{L}_2\cap\mathcal{L}_\infty$, $L_i\in\mathcal{L}_2\cap\mathcal{L}_\infty$. Thus by (\ref{eq:dem}-\ref{eq:des}), one has that $\dot{e}_m, \dot{e}_s\in\mathcal{L}_\infty$. Now invoking Barbalat's Lemma, we conclude that $\lim_{t\rightarrow\infty}e_i(t)=0$.   Similarly, we have that $\lim_{t\rightarrow\infty}L_i(t)=0$ by Barbalat's Lemma. Thus, by (\ref{eq:e}), we arrive at $\lim_{t\rightarrow\infty}e(t)=0$.

Furthermore, the boundedness of $q_i, \dot{q}_i, e_i$ implies that $\eta_i,\tau_i\in\mathcal{L}_\infty$. Thus by the dynamic model (\ref{eq:master}-\ref{eq:slave}) and the Properties P1, P3, we otbain that $\ddot{q}_i\in\mathcal{L}_\infty$. Hence involing Barbalat's Lemma again, we arrive at that $\lim_{t\rightarrow\infty}\dot{q}_i(t)=0$.

Now we show that the parameter estimation error $\tilde{\theta}_i$ approaches to zero as $t\rightarrow \infty$. Note that the parameter adaption law (\ref{eq:updatelaw_m}) implies that
\[\dot{\tilde{\theta}}_i=-\Gamma_i(Y_i^T\eta_i+(\xi_i+\delta_i)Pi\tilde{\theta}_i)\in\mathcal{L}_\infty\]
Similarly, the conclusion that $\lim_{t\rightarrow \infty}\tilde{\theta}_i(t)=0$ is guaranteed by using  Barbalat's Lemma.
By now, Claim 1 is established.

To illustrate the transient performance of the teleoperators, we start from the convergence of estimation errors $\tilde{\theta}_i$. Obviously, Let $V_\theta(t)\triangleq\sum_{i=m,s}\frac{2}{k_i\lambda_i}\tilde{\theta}_i^T\Gamma_i^{-1}\tilde{\theta}_i$. The time derivative of $V_\theta$ is given by
\begin{align*}
\dot{V}_\theta(t)=&\sum_{i=m,s}\frac{2}{k_i\lambda_i}\tilde{\theta}_i^T\Gamma_i^{-1}\dot{\tilde{\theta}}_i\\
\leq &\sum_{i=m,s}\frac{2}{k_i\lambda_i}(-\tilde{\theta}_i^TY_i^T\eta_i-\tilde{\theta}_i^T\alpha_i\|Y_i^T\eta_i\|\tilde{\theta}_i-\delta_i\kappa_{i0}\tilde{\theta}_i^T\tilde{\theta}_i)\\
\leq &\sum_{i=m,s}\frac{2}{k_i\lambda_i}((1-\alpha_i\|\tilde{\theta}_i\|)\|Y_i^T\eta_i\|\|\tilde{\theta}_i\|-\delta_i\kappa_{i0}\|\tilde{\theta}_i\|^2)
\end{align*}
Thus if $\|\theta_i\|\geq 1/\alpha_i$, we have $\dot{V}_\theta(t)\leq -\sum_{i=m,s}\frac{2\delta_i\kappa_{i0}}{k_i\lambda_i}\|\tilde{\theta}_i\|^2$. This implies that $\dot{V}_\theta$ is always negative when $\|\theta\|\geq \sqrt{2}/\alpha$ with $\alpha=\min\{\alpha_m,\alpha_s\}$.
So the parameter error $\tilde{\theta}(t)=\text{col}\{\tilde{\theta}_m, \tilde{\theta}_s\}$ will converge to a sphere $\Omega_\theta=\{\tilde{\theta}:\|\tilde{\theta}\|\leq \sqrt{2\lambda_M^\theta\mu_M^\theta/\lambda_m^\theta\mu_m^\theta}/\alpha\}$,
where $\lambda_m^\theta=\min\{\lambda_{\min}(\Gamma_m^{-1}), \lambda_{\min}(\Gamma_s^{-1})\} $, $\lambda_M^\theta=\max\{\lambda_{\max}(\Gamma_m^{-1}), $ $\lambda_{\max}(\Gamma_s^{-1})\}$,$\mu_m^\theta=\min\{\frac{1}{k_m\lambda_m}, \frac{1}{k_s\lambda_s}\}$, $\mu_M^\theta=\max\{\frac{1}{k_m\lambda_m},$ $ \frac{1}{k_s\lambda_s}\}$ within a given time.  The proof of Claim 2 is concluded.

\end{proof}

\begin{remark}
Compared to the existing works \cite{chopra2008brief, Nuno2010improvedsyn, SARRAS2014_JFI}, the proposed control scheme guarantees the convergence of parameters to their true values, while neither the PE or SE condition is required. This is accomplished by the boundedness of the matrix $P_i$ in the new-defined prediction error $z_i$.
\end{remark}

When the external forces are not zero, we have the following result.
\begin{proposition}
Consider the bilateral teleoperation system (\ref{eq:master}-\ref{eq:slave}) controlled by (\ref{eq:taum}-\ref{eq:taus}) together with the updating law (\ref{eq:updatelaw_m}-\ref{eq:update_lambda}) under the communication channel satisfying Assumption \ref{amp:delay}, if there exist positive-definite matrices $R_{m},R_{s}$ such that the following LMI holds, respectively:
\begin{eqnarray}
\Xi=\begin{bmatrix}\Xi_1 & 0 &0 & -I\\
*& \Xi_2 &-I & 0\\
* & *& -\frac{R_m}{h_m} & 0\\
* & * & * &-\frac{R_s}{h_s} \end{bmatrix}<0,\label{eq:LMI_PI}
\end{eqnarray}
with
\begin{align*}
\Xi_1=-\frac{1}{\lambda_m}I+h_mR_m+\frac{2\lambda_s(\rho_s^M)^2d_m^2}{(1-d_m)k_s^2}I, \\
\Xi_2=-\frac{1}{\lambda_s}I+h_sR_s+\frac{2\lambda_m(\rho_m^M)^2d_s^2}{(1-d_s)k_m^2}I,
\end{align*}
then  when the external forces satisfy that  $F_h, F_e\in\mathcal{L}_\infty\cap\mathcal{L}_2$,  all the signals are bounded and the position errors, velocities and the estimation errors asymptotically converge to zero, that is, $\lim_{t\rightarrow \infty} q_m-q_s=\lim_{t\rightarrow \infty}\dot{q}_m=\lim_{t\rightarrow \infty}\dot{q}_s=\lim_{t\rightarrow \infty}\tilde{\theta}_m=\lim_{t\rightarrow \infty}\tilde{\theta}_s=0$. Moreover, the estimation errors converge into a specified domain within a given time.
\end{proposition}
\begin{proof}
Choose the Lyapunov-functional candidate $\bar{V}$ as follows:
\begin{equation}\bar{V}=2V_1+V_2+V_3+2V_4\label{eq:barV}\end{equation}where $V_1,V_2, V_3, V_4$ are defined in (\ref{eq:V1}), (\ref{eq:V2}), (\ref{eq:V3}), and (\ref{eq:V4}), respectively.Thus we have
\begin{eqnarray*}\dot{\bar{V}}&\leq&-\xi\Xi\xi+\sum_{i=m,s} (-\frac{4\delta_i\kappa_{i0}}{k_i\lambda_i}|\tilde{\theta}_i|^2-\lambda_ie_i^Te_i+\frac{4}{k_i^2\lambda_i}\|F_i\|^2)\label{eq:dbarV}\\
\end{eqnarray*}
Integrating both sides of  (\ref{eq:dbarV}) from $0$ to $\infty$, we have
\[\bar{V}(\infty)-\bar{V}(0)\leq \frac{4}{k_i\lambda_i}\int_{0}^\infty\|F_i\|^2\]
which implies that $\bar{V}_\infty\in\mathcal{L}_\infty$. Thus following the same line of reasoning in the proof of Theorem \ref{Thm:mainresult}, we  can obtain the conclusion.
\end{proof}

\section{New tracking performance measures}\label{sec:trackingindex}
In this section, we propose new tracking performance measures for bilateral teleoperation systems.  A requirement for bilateral teleoperation is that the slave should follow the master's motion. Specifically, when the slave is in free motion, there is no environmental force between the slave and the environment, and hence the slave should follow the master' motion very tightly.  However, in terms of the tracking performance, it should be noted that the tracking performance is related to the slave's desired position, i.e., the master's position, and it's actual position, i.e., the slave's position. Hence it is naturally to define a ratio between the position error and the slave's desired position as the position tracking performance measure $\Delta_p$, i.e.,
\begin{align}&\Delta_{p}^i=\frac{|q_m^i-q_s^i|}{|q_m^i|}, q_m^i\neq 0 \end{align}
where $i=1,...,n$ representing the $i$-th joint. For simplicity, we assume that $\Delta_p^i=0$ when $q_m^i=0$.
 Similarly, when there is contact between the environment and the slave robot, the force tracking performance is  related a ratio between the force error and the contact force, that is,
 \begin{align}\Delta_f^i=\frac{|F_m^i-F_s^i|}{|F_s^i|}|, F_s^i\neq 0.
\end{align}
where $i=1,...,n$, and $\Delta_f^i=0$ when $F_s^i=0$.

In the presence of communication, the measures  $\Delta_p, \Delta_f$ should be modified as
\begin{eqnarray}
\Delta_p^i=\left\{\begin{aligned}\frac{|q_m^i(t-T_m(t))-q_s^i|}{q_m^i(t-T_m(t))}, (q_m^i(t-T_m(t))\neq 0)\\0, (q_m^i(t-T_m(t))\neq 0)\end{aligned}\right.\\
\Delta_f^i=\left\{\begin{aligned}\frac{|F_m^i-F_s^i(t-T_s(t))|}{|F_s^i(t-T_s)|}, (F_s^i(t-T_s(t))\neq 0)\\0, (F_s^i(t-T_s(t))=0)\end{aligned}\right.
\end{eqnarray}
where $i=1,...,n$.
%

In summary , we define performance measure indexes for each joint
\begin{eqnarray}\Delta_{J_p}^i=\int_{0}^\infty|\Delta_p^i(t)|dt,\\\Delta_{J_f}^i=\int_{0}^\infty|\Delta_f^i(t)|dt\end{eqnarray} Obviously, the smaller the indexes $\Delta_{J_p}^i, \Delta_{J_f}^i$ are, the better the tracking performance is for each joint.
\section{Simulations}\label{sec:simulations}
In this section, the simulation results are shown to verify the effective of the main results.

\subsection{Simulation Setup}Consider a 2-DOF teleoperation system with the following dynamics
\begin{align}
&M_m\!(q_m\!)\ddot{q}_m\!+\!C_m\!(q_m,\dot{q}_m\!)\dot{q}_m\!+\!G_m\!(q_m\!)\!=\!J_m^TF_h\!+\!\tau_m\\
&M_{s}(q_{s})\ddot{q}_{s}\!+\!C_{s}(q_{s},\dot{q}_{s})\dot{q}_{s}\!+\!G_{s}(q_{s})=J_s^TF_{e}\!+\!\tau_{s}
\end{align}
where
\begin{align*}
&M_i(q_i)=\begin{bmatrix}
M_{i_{11}}(q_i)&M_{i_{12}}(q_i)\\
M_{i_{21}}(q_i)&M_{i_{22}}(q_i)
\end{bmatrix},\\
&C_i(q_i,\dot{q}_i)=\begin{bmatrix}
C_{i_{11}}(q_i,\dot{q}_i)&C_{i_{12}}(q_i,\dot{q}_i)\\
C_{i_{21}}(q_i,\dot{q}_i)&C_{i_{22}}(q_i,\dot{q}_i)
\end{bmatrix},\\
&G_i(q_i)=\begin{bmatrix}G_{i_1}\\G_{i_2}\end{bmatrix},
\end{align*}
for $i=m, s$, respectively, and
\begin{equation*}
\begin{aligned}
&M_{i_{11}}\!(q_i\!)\!=\!l_{i_2}^2m_{i_2}\!+\!l_{i_1}^2(m_{i_1}\!+\!m_{i_2})\!+\!2l_{i_1}l_{i_2}m_{i_2}\cos\!(q_{i_2}\!),\\
&M_{i_{22}}(q_i)\!=\!l_{i_2}^2m_{i_2},\\
&M_{i_{12}}(q_i)\!=\!M_{i_{21}}(q_i)=l_{i_2}^2m_{i_2}+l_{i_1}l_{i_2}m_{i_2}\cos(q_{i_2}),\\
&C_{i_{11}}(q_i,\dot{q}_i)\!=\!-l_{i_1}l_{i_2}m_{i_2}\sin(q_{2_i})\dot{q}_{i_2},
\\&C_{i_{12}}(q_i,\dot{q}_i)\!=\!-l_{i_1}l_{i_2}m_{i_2}\sin(q_{2_i})(\dot{q}_{1_i}+\dot{q}_{2_i}),&\\
&C_{i_{21}}(q_i,\dot{q}_i)\!=\!l_{i_1}l_{i_2}m_{i_2}\sin(q_{2_i})\dot{q}_{1_i},C_{i_{22}}(q_i,\dot{q}_i)=0,\\
&G_{i_1}(q_i)\!=\!\frac{1}{l_{i_2}}gl_{i_2}^2m_{i_2}\cos(q_{1_i}+q_{2_i})\\&\quad\quad\quad\quad+
\frac{1}{l_{i_1}}(l_{i_2}^2m_{i_2}+l_{i_1}^2(m_{i_1}+m_{i_2})\\&\quad\quad\quad\quad-l_{i_2}^2m_{i_2})\cos(q_{1_i}),&\\
&G_{i_2}(q_i)\!=\!\frac{1}{l_{i_2}}gl_{i_2}^2m_{i_2}\cos(q_{1_i}+q_{2_i}).&
\end{aligned}
\end{equation*}
\begin{figure}[h!]
\begin{center}
\includegraphics[width=0.8\linewidth]{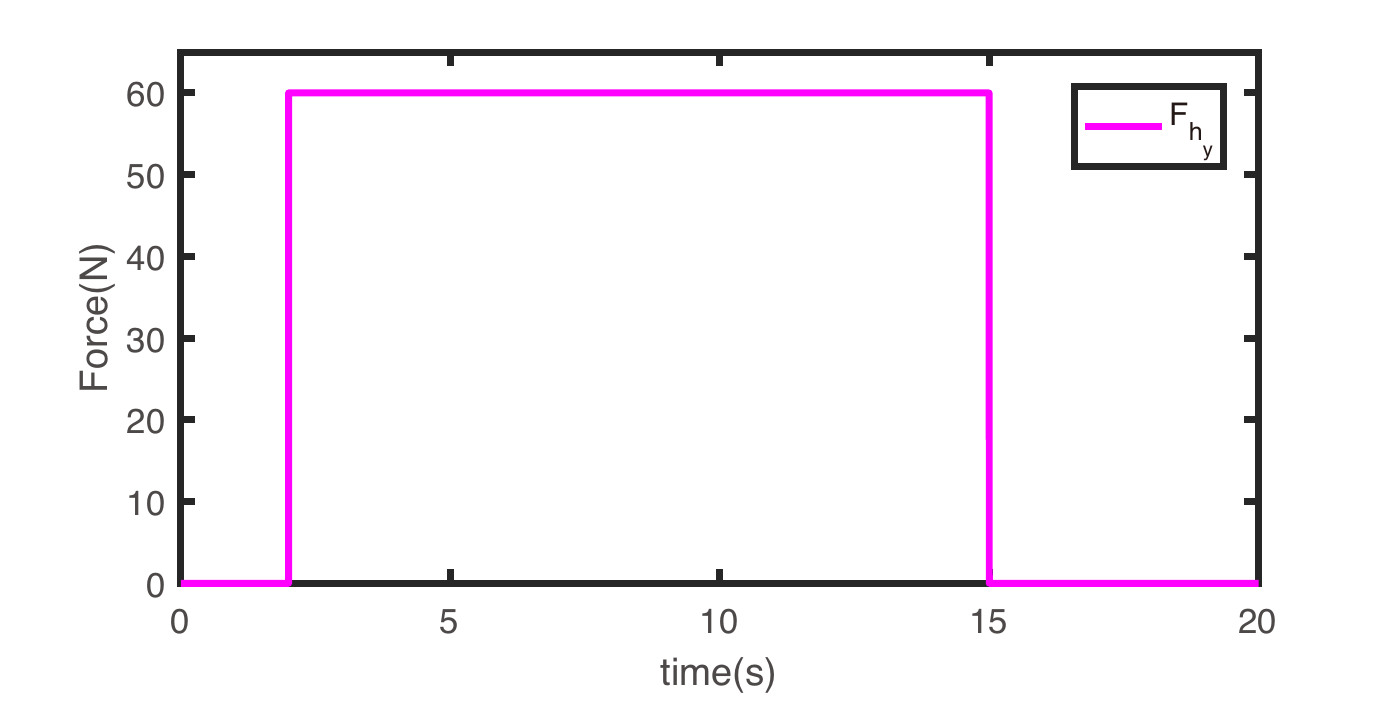}
\caption{The rectangle human input force}
\label{fig:fhy_free}
\end{center}
\end{figure}
The mass of the manipulators are chosen as $m_{m_1}=1.5$kg, $m_{m_2}=0.75$kg, $m_{s_1}=2.5$kg, $m_{s_2}=1.5$kg, the length of links for the master and the slave robots are $l_{m_1}=l_{s_1}=0.5$m, $l_{m_2}=l_{s_2}=0.3$m. The Jacobians of the master and slave robots are given by
\begin{align}&J_i(q_i)=\begin{bmatrix}\!-\!l_{i_1}\sin(q_{i_1})\!-\!l_{i_2}\sin(q_{i_1}\!+\!q_{i_2})& \!-\!l_{i_2}\sin(q_{i_1}\!+\!q_{i_2})\\l_{i_1}\cos(q_{i_1})\!+\!l_{i_2}\cos(q_{i_1}\!+\!q_{i_2})&l_{i_2}\cos(q_{i_1}\!+\!q_{i_2})
\end{bmatrix}.\end{align}

The following parameterization is proposed for both manipulators with $i=m, s$, respectively:
\begin{align*}
&Y_i(q_i,\dot{q}_i,\ddot{q}_i)\\&=\begin{bmatrix} \ddot{q}_{i_1}&Y_{12}&\ddot{q}_{i_2}&g\cos(q_{i_1}\!+\!q_{i_2})&g\cos(q_{i_1})\\0&Y_{22}&\ddot{q}_{i_1}\!+\!\ddot{q}_{i_2}&g\cos(q_{i_1}\!+\!q_{i_2})&0\end{bmatrix},\\
&Y_{12}\!=\!2\!\cos\!(q_{i_2})\!\ddot{q}_{i_1}\!+\!\cos\!(q_{i_2})\!\ddot{q}_{i_2}\!-\!2\!\sin\!(q_{i_2}\!)\dot{q}_{i_1}\dot{q}_{i_2}\!-\!\sin\!(q_{i_2}\!)\dot{q}_{i_2}^2,\\
&Y_{22}=\cos(q_{i_2})\ddot{q}_{i_1}+\sin(q_{i_2})\dot{q}_{i_1}^2,\\
&\theta_i=col\{\theta_{i_1},\theta_{i_2},\theta_{i_3},\theta_{i_4},\theta_{i_5}\},
\end{align*}
where $\theta_{i_1}=l_{i_2}^2m_{i_2}+l_{i_1}^2(m_{i_1}+m_{i_2})$, $\theta_{i_1}=l_{i_2}^2m_{i_2}+l_{i_1}^2(m_{i_1}+m_{i_2})$, $\theta_{i_3}=l_{i_2}^2m_{i_2}$, $\theta_{i_4}=l_{i_2}m_{i_2}$, $\theta_{i_5}=l_{i_1}(m_{i_1}+m_{i_2})$.

\begin{figure}
\centering
\subfigure[The first joint]{\includegraphics[width=0.8\linewidth]{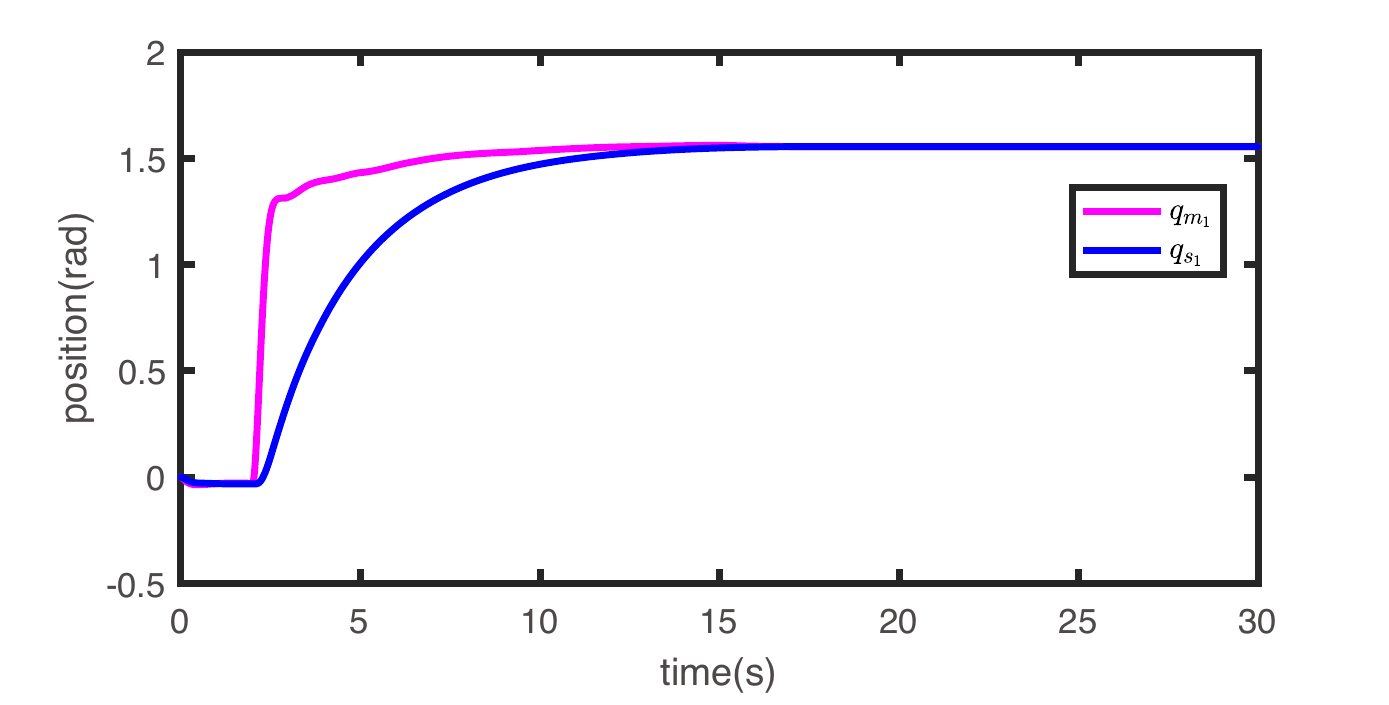}\label{fig:pos_link1_free}}
\subfigure[The second joint]{\includegraphics[width=0.8\linewidth]{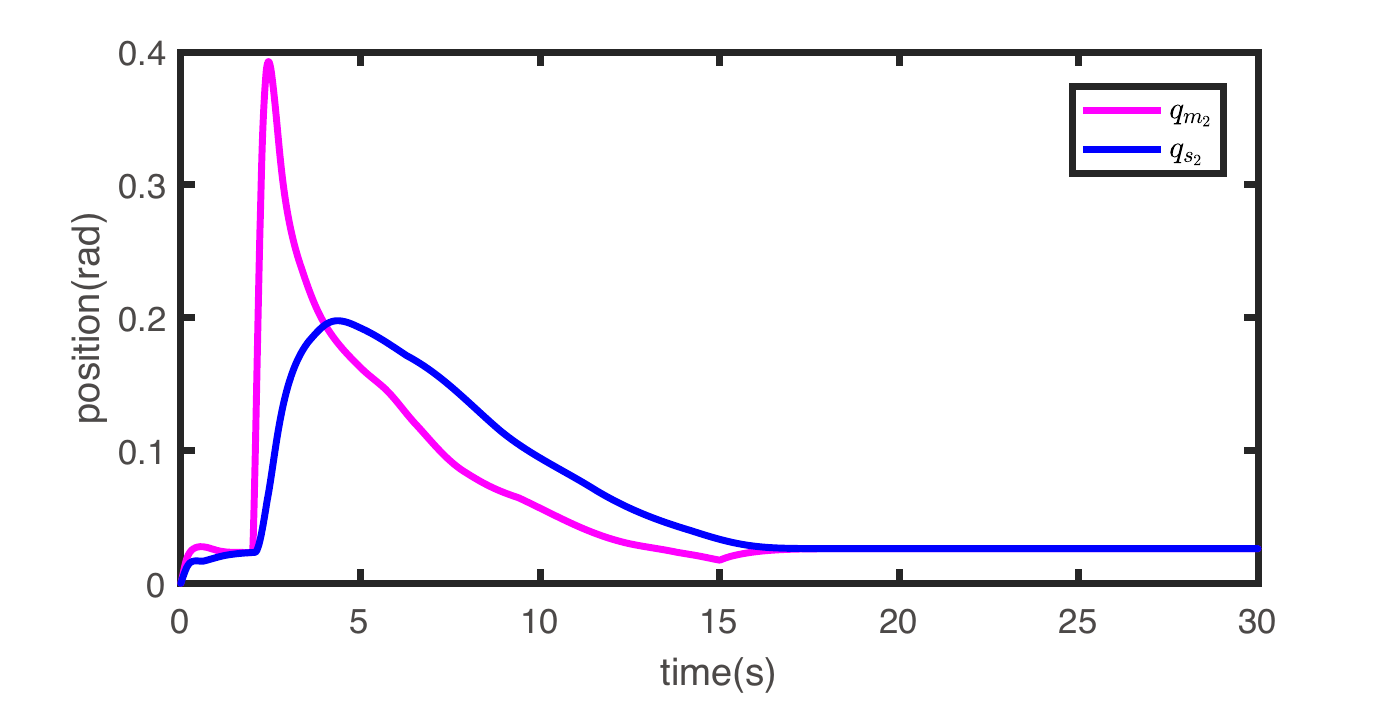}}
\caption{Scenario A: the master and slave robots joint positions}
\label{fig:pos_link_free}
\end{figure}

\begin{figure}[!h]
\centering
\subfigure[The first joint]{\includegraphics[width=0.8\linewidth]{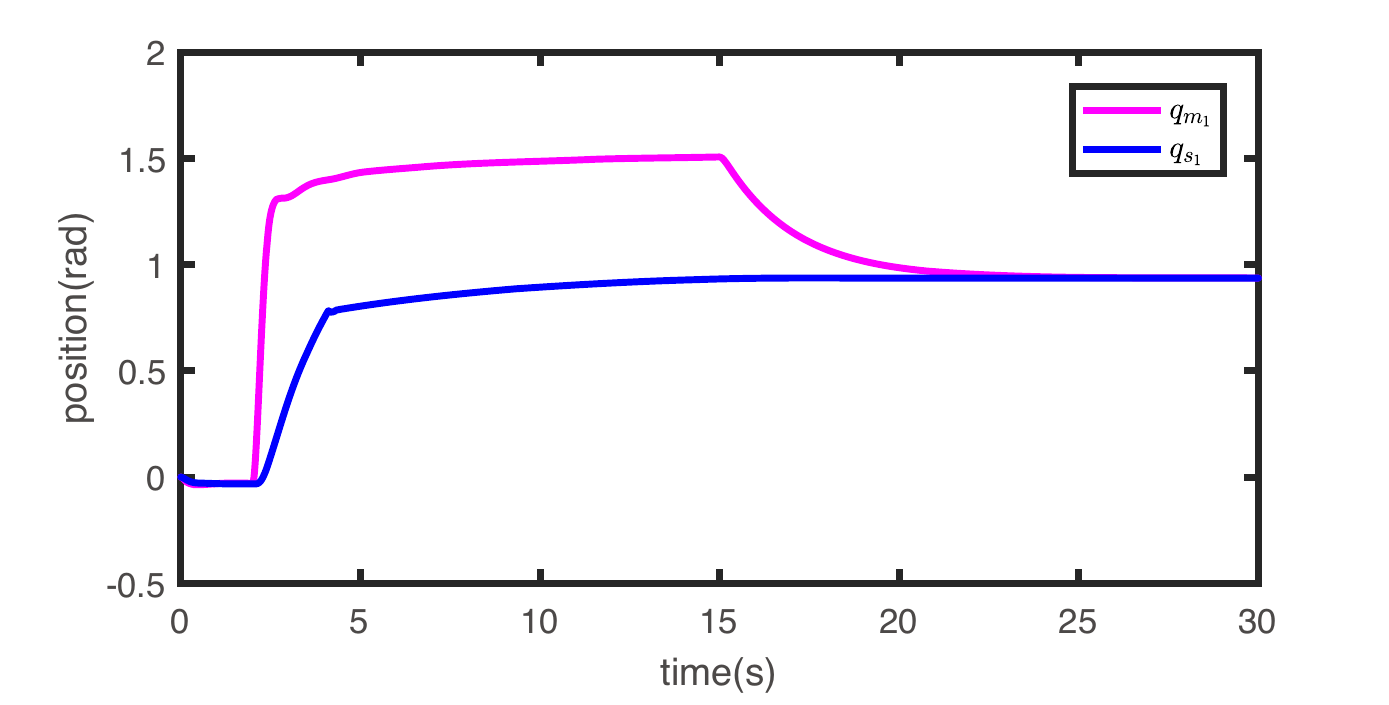}}
\subfigure[The second joint]{\includegraphics[width=0.8\linewidth]{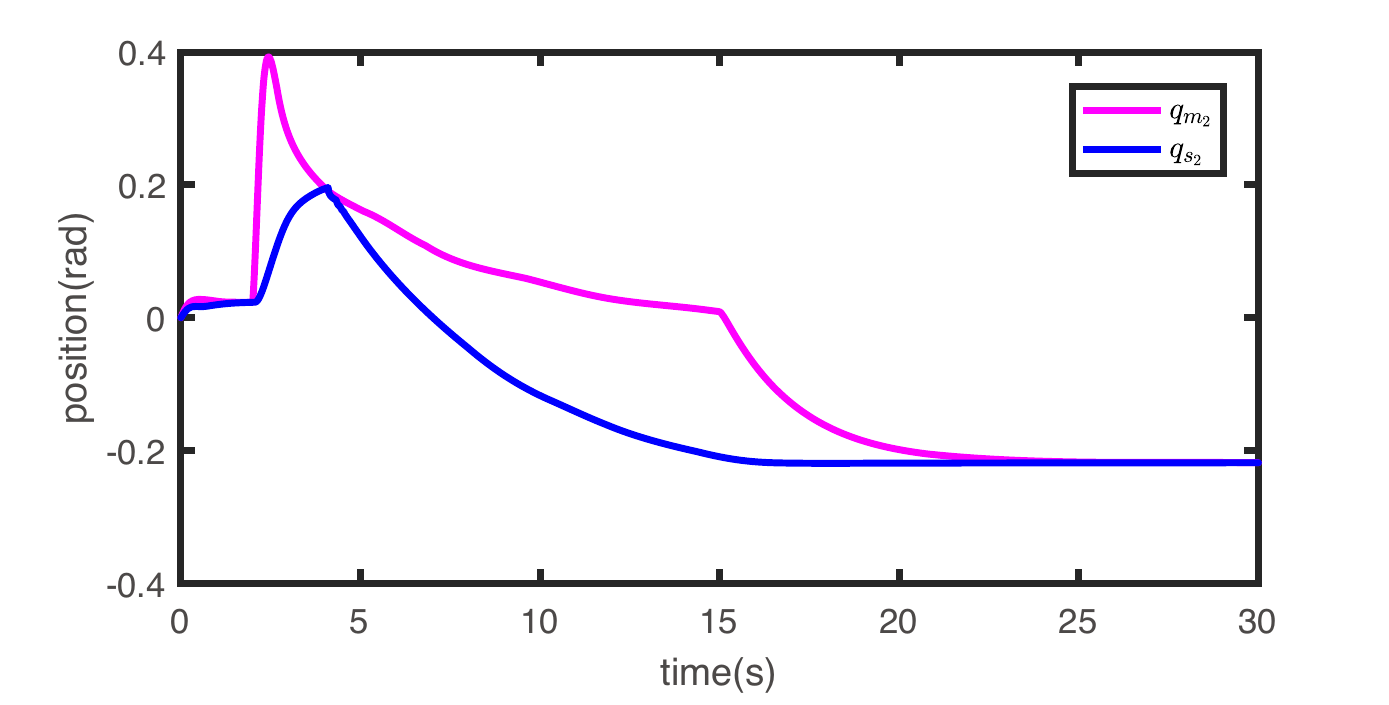}}
\caption{Senario B: the master and slave robots joint positions}
\label{fig:pos_link_contact}
\end{figure}

We assume that the operator hand force at the $Y-$direction is generated by a step signal depicted in Figure~\ref{fig:fhy_free},  while at the $X$-direction, there is no external force, then we have $F_h=[0,1]^TF_{h_y}$.

\begin{figure}[!h]
\centering
\subfigure[x]{\includegraphics[width=0.8\linewidth]{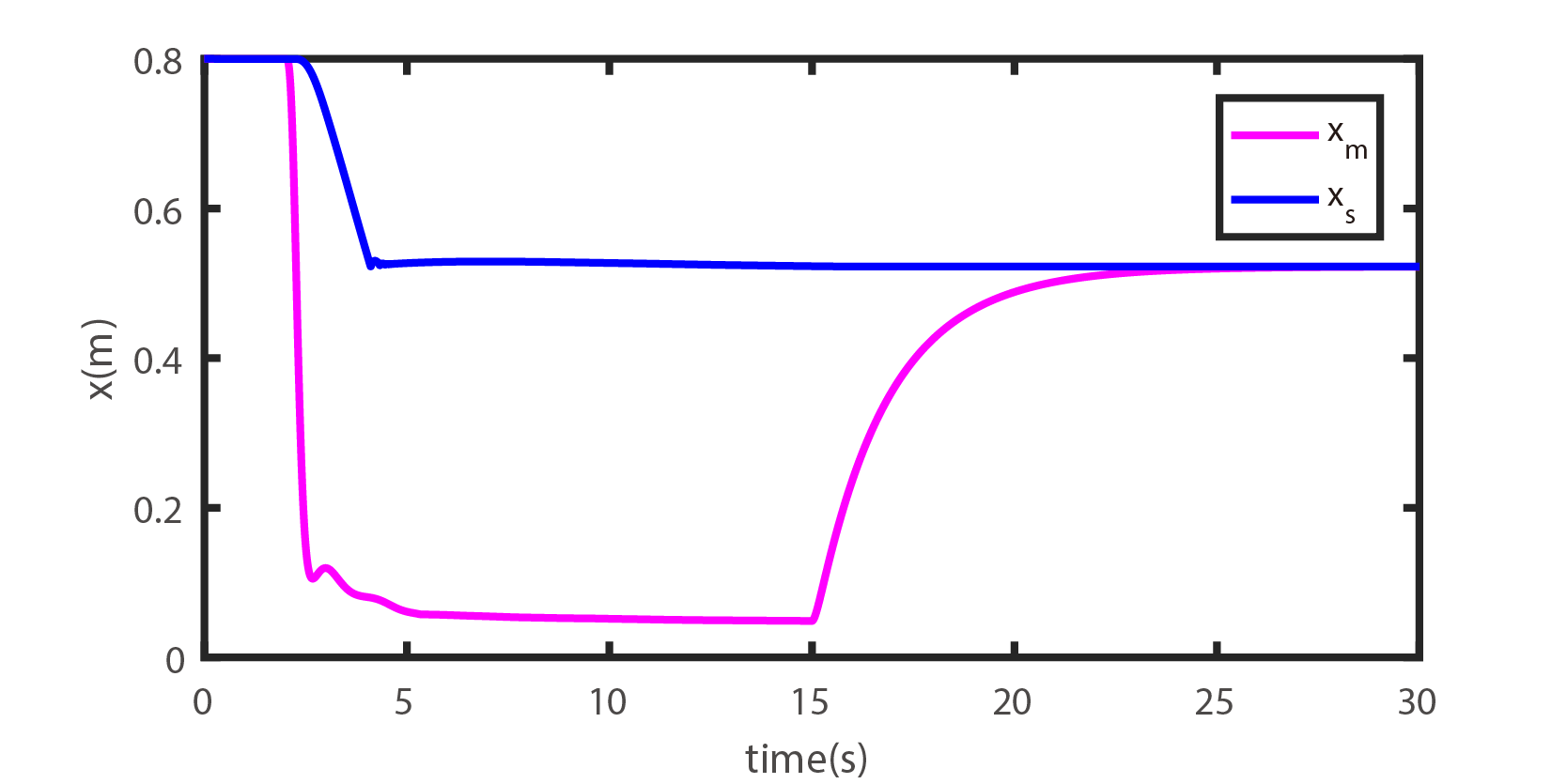}}
\subfigure[y]{\includegraphics[width=0.8\linewidth]{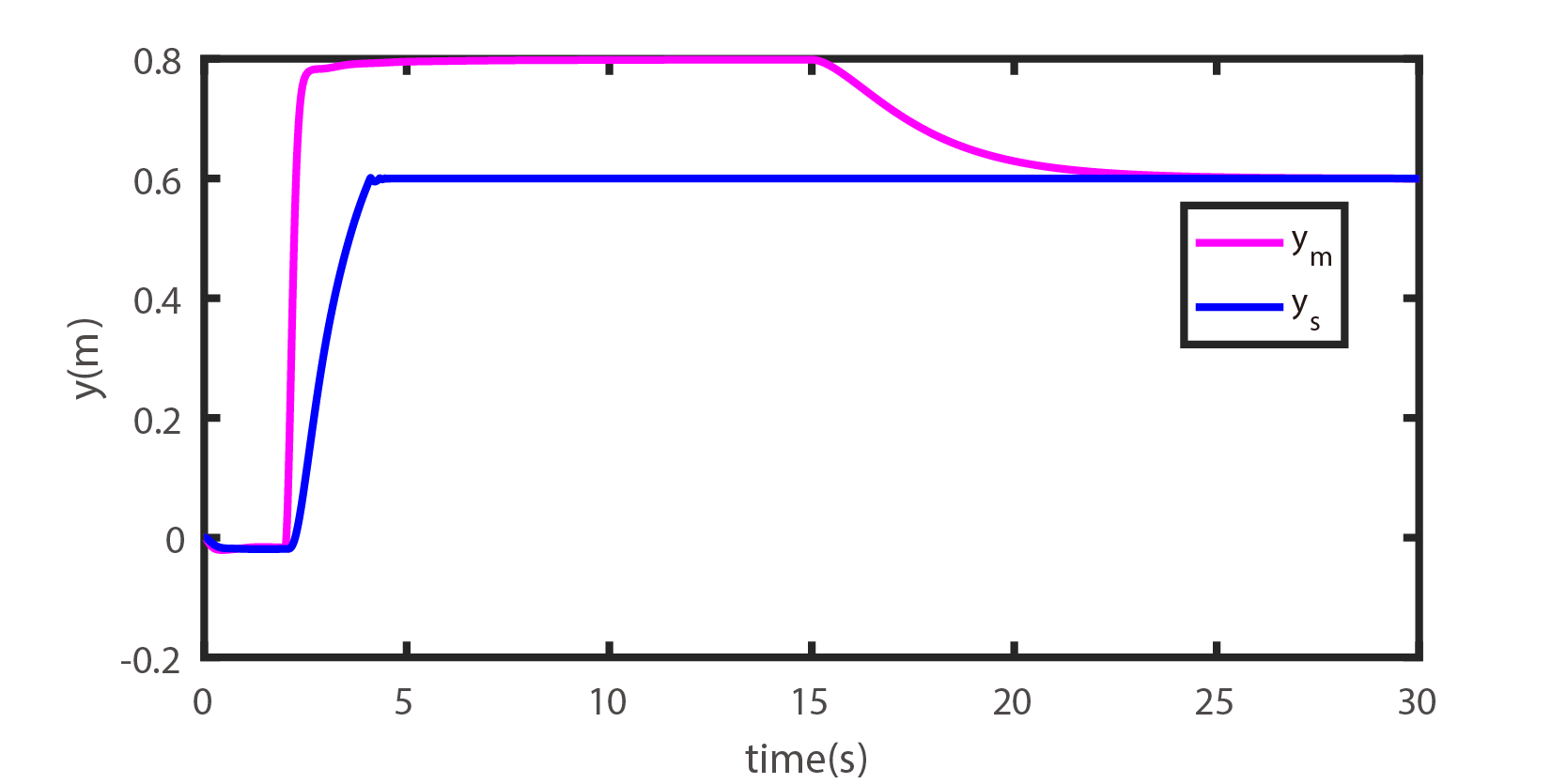}}
\caption{Senario B: The master and slave robots end-effector positions}
\label{fig:pos_xy_contact}
\end{figure}

In the slave side, two scenarios are considered. The first one which we called as Scenario A is that the slave manipulator is in free motion.  The second one which is named as Scenario B is that there is a wall at $y=0.6$ m in the environment. When the slave end-effector reaches the wall and moves further, the interaction force is $10,000(y-0.6)$ N. It can be seen that the wall is very stiff. In both scenarios, the master and the slave start with zero initial conditions, i.e., $q_m(0)=q_s(0)=[0,0]^T$ rad, $\dot{q}_m(0)=\dot{q}_s(0)=[0,0]^T $ rad/s. The initial values for the estimated dynamic parameters are chosen as $\hat{\theta}_m=[0.4, 0.1, 0.2, 0.32, 0.7]^T$, $\hat{\theta}_s=[0.7, 0.2, 0.3, 0.5, 1.7]^T$. The communication delays are set as: $T_m(t)=0.3+0.2\sin(2t)+0.1\sin(5t)$, $T_s(t)=0.8+0.3\sin(1.5t)+0.1\sin(5t)$.
\subsection{Stability Verification}

First of all, the slave is in free motion.  By applying the designed controller (\ref{eq:taum}), (\ref{eq:taus}), (\ref{eq:updatelaw_m}), (\ref{eq:updatelaw_z}), (\ref{eq:update_P}), (\ref{eq:update_lambda}) with $K_m=K_s=100I, \lambda_m=\lambda_s=0.5$, we obtain the simulation results as shown in Figure~\ref{fig:pos_link_free}. It can be seen that under the proposed controller, the presence of parametric uncertainties does not violate the stability of the bilateral teleoperation. The master and the slave achieve synchronization around the time $t=10$s.  Furthermore, the estimated dynamic parameters  are shown in  Figure~\ref{fig:theta_m_free_CAC}and  Figure~\ref{fig:theta_s_free_CAC}, respectively. By Figure~\ref{fig:theta_m_free_CAC}, it is easy to find that the estimate  $\hat{\theta}_m$ converges to its real value $\theta_m$ after $t=2$ s, when the external human force starts to be exerted to the master manipulator. Similarly, Figure \ref{fig:theta_s_free_CAC} reveals the convergence of the slave's dynamic parameters to their true values.

\begin{figure}[h!]
\centering
\includegraphics[width=0.8\linewidth]{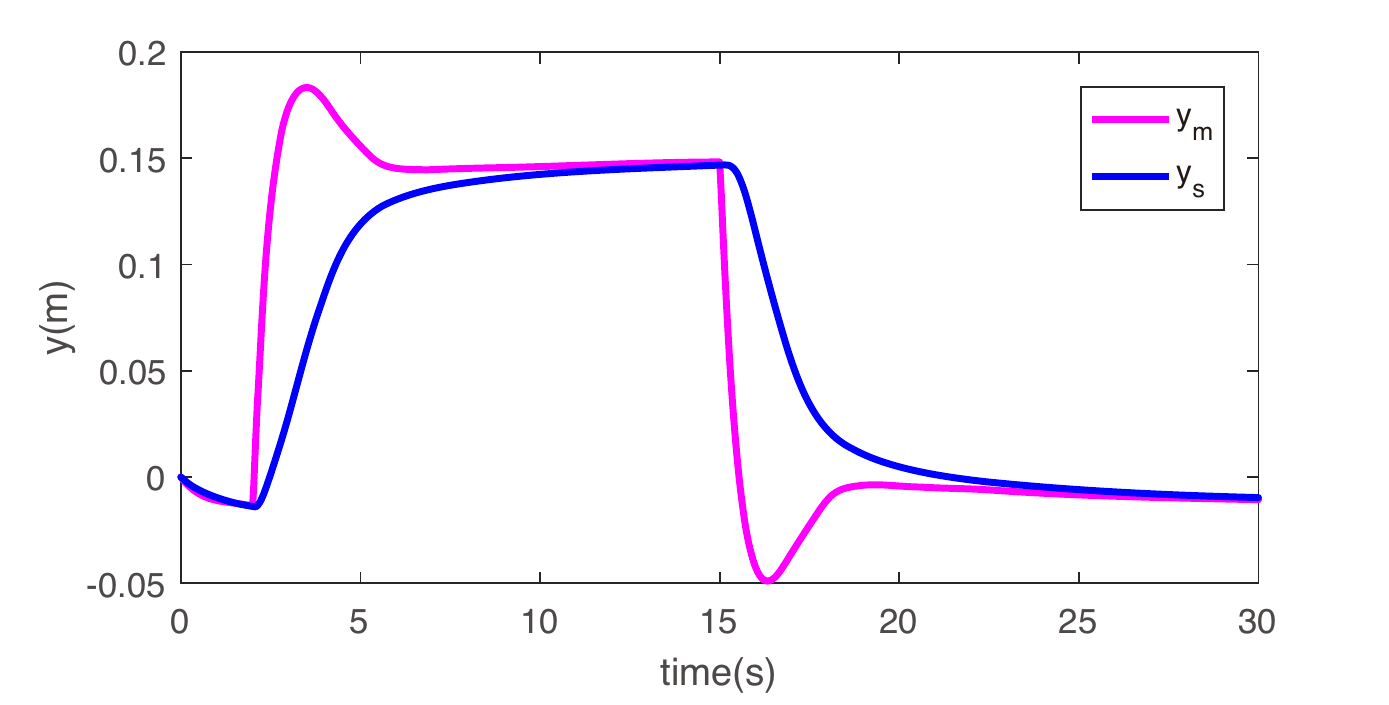}
\caption{Scenario A: The master and slave robots end-effector positions by using the method in \cite{SARRAS2014_JFI}}
\label{fig:noparaconverge_pos_y_free}
\end{figure}

Secondly, we show the tracking performance when the slave would contact a wall at $y=0.6$ m. From Figure~\ref{fig:pos_link_contact}-Figure~\ref{fig:pos_xy_contact}, it is easy to find that the system is stable even there exist dynamic parameter uncertainties. The master and the slave move forward when the external human force takes action and stop moving when the slave contact the wall at Y-direction.  The master's positions gradually converge to the slave's positions when the human operator stops exerting forces to the considered teleoperation system.

\begin{figure}[h!]
\centering
\subfigure[The proposed method]{\includegraphics[width=0.8\linewidth]{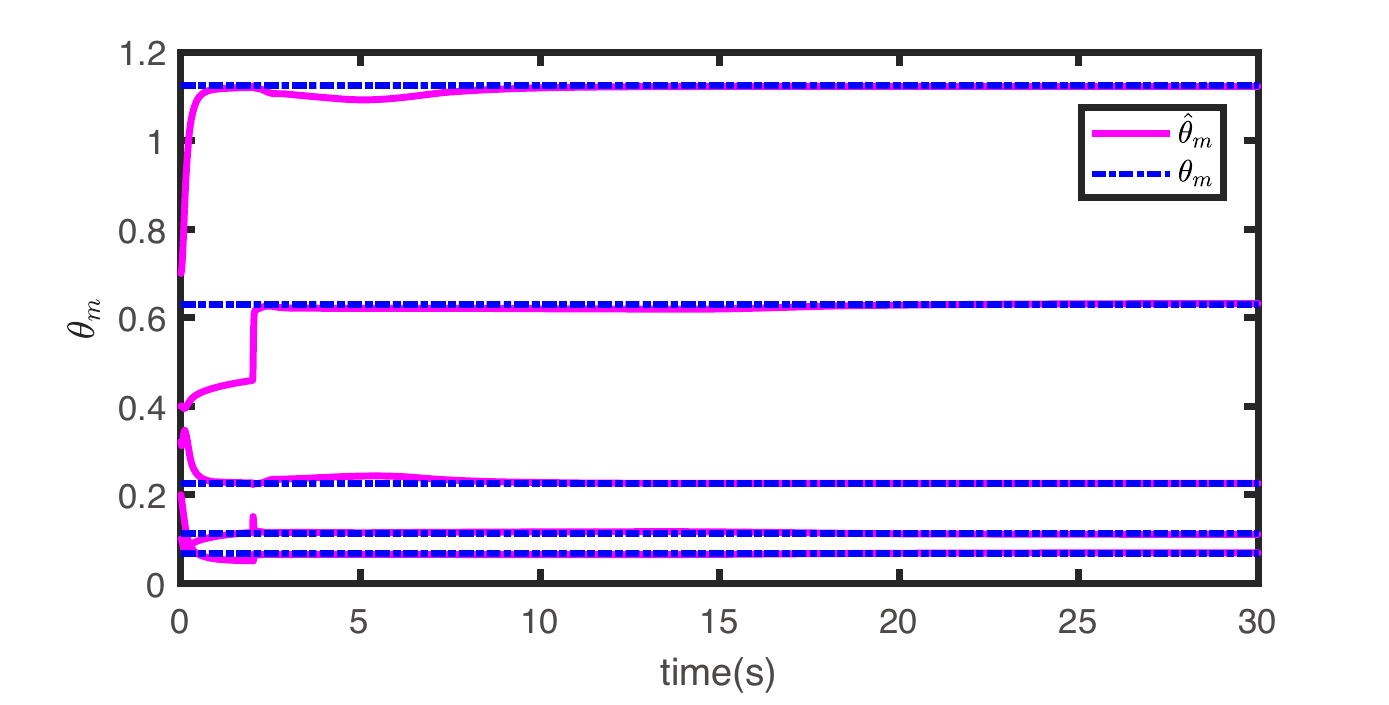}\label{fig:theta_m_free_CAC}}
\subfigure[The method in \cite{SARRAS2014_JFI}]{\includegraphics[width=0.8\linewidth]{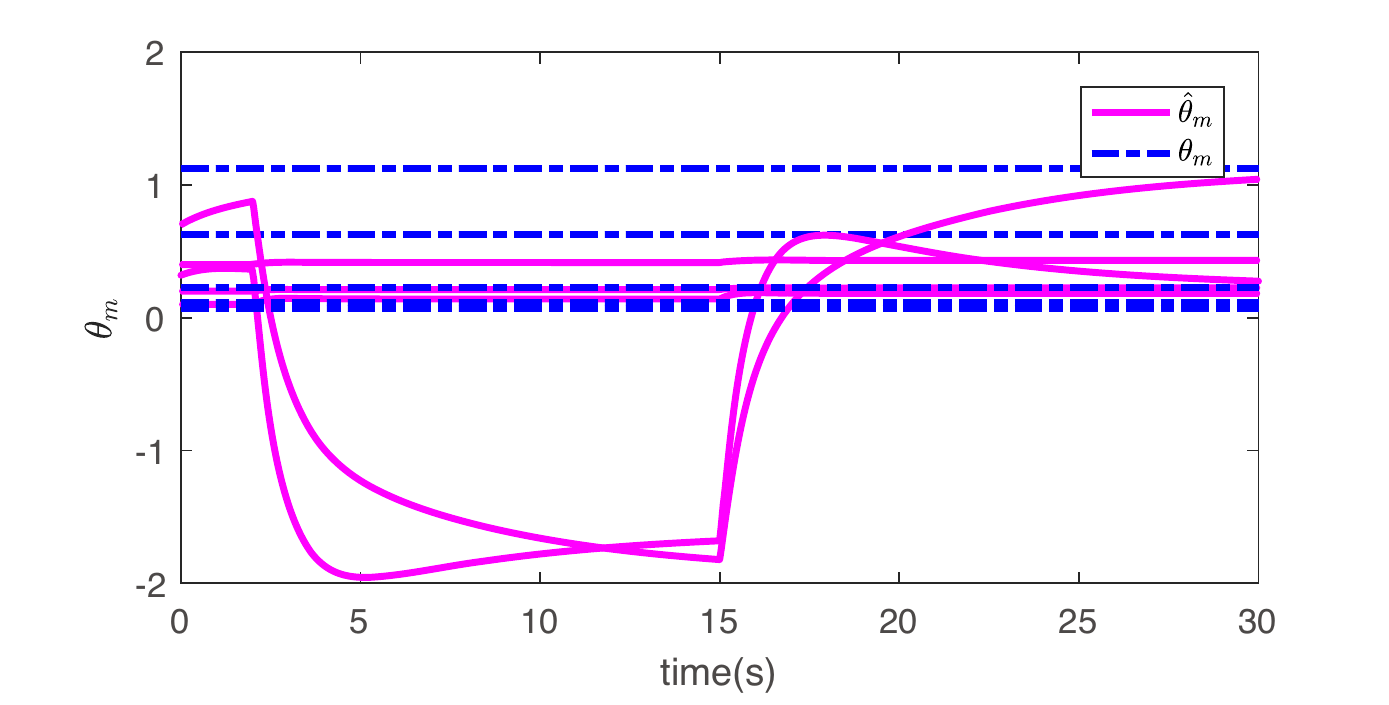}\label{fig:theta_m_free_17}}
\caption{Scenario A: The dynamic parameter $\theta_m$ and the parameter estimate $\hat{\theta}_m$}
\label{fig:theta_m_free}
\end{figure}

\begin{figure}[h!]
\centering
\subfigure[The proposed method]{\includegraphics[width=0.8\linewidth]{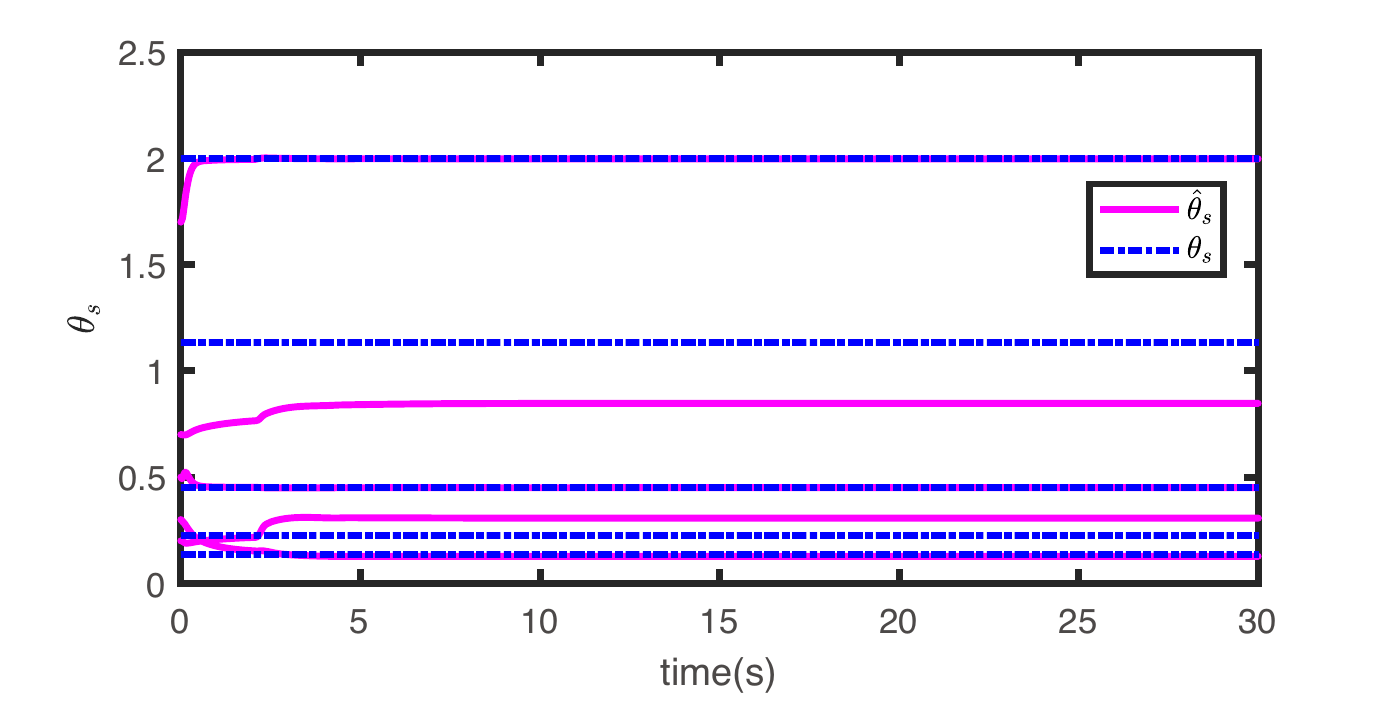}\label{fig:theta_s_free_CAC}}
\subfigure[The method in \cite{SARRAS2014_JFI}]{\includegraphics[width=0.8\linewidth]{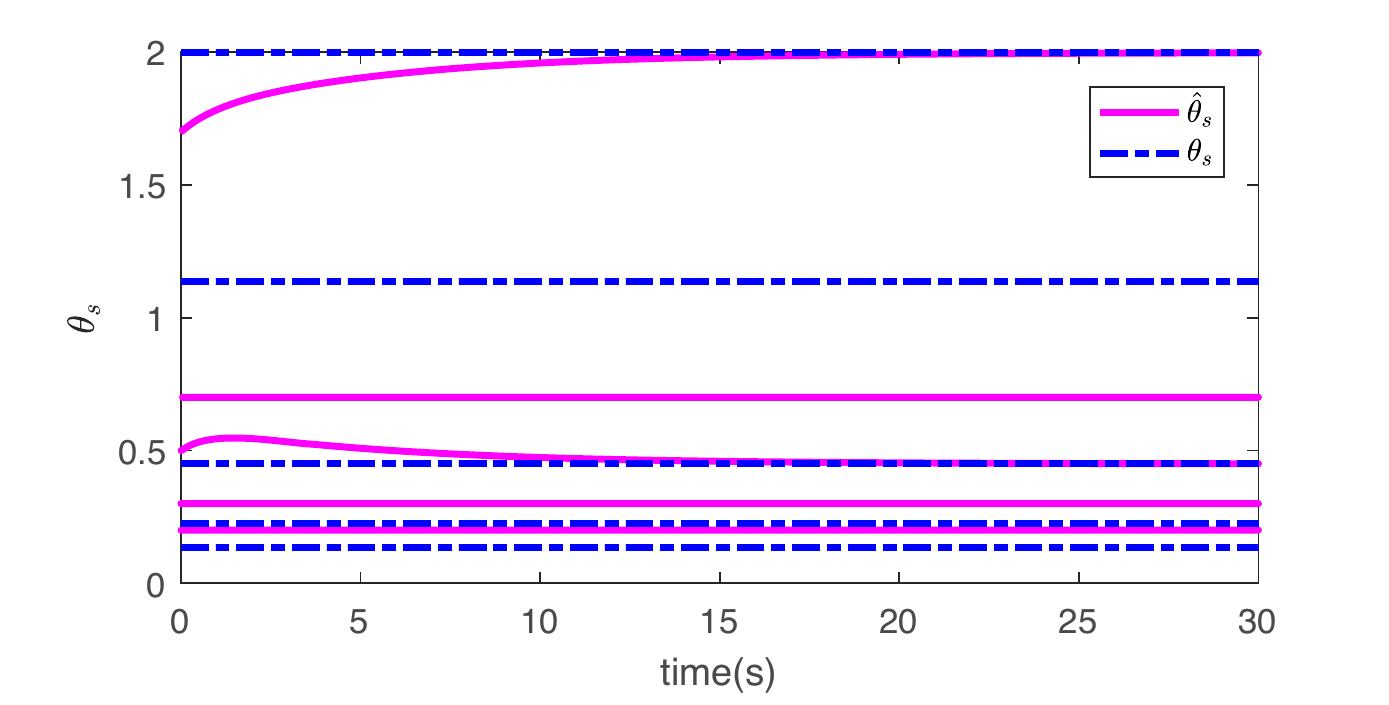}\label{fig:theta_s_free_17}}
\caption{Scenario A: The dynamic parameter $\theta_s$ and the parameter estimate $\hat{\theta}_s$}
\label{fig:theta_s_free}
\end{figure}

\subsection{Comparison study between this paper and \cite{SARRAS2014_JFI}}
In order to show the effectiveness of the proposed method, comparative simulation studies with the method in \cite{SARRAS2014_JFI} is given here. Both of the two methods can achieve asymptotic position tracking under time-varying time delays. However, the results in \cite{SARRAS2014_JFI} did not consider the parameter convergence. The dynamic parameter estimates in Scenario A under the proposed method in this paper and the method in \cite{SARRAS2014_JFI} are depicted in Figure~\ref{fig:theta_m_free} and Figure~\ref{fig:theta_s_free}.  It is shown that the master's dynamic parameters converge to the true values very quickly under the proposed method while in contrary the estimates in \cite{SARRAS2014_JFI} diverge from the true values almost all the time. The slave's dynamic estimates also have a better tracking performance with its true values under our proposed method from Figure ~\ref{fig:theta_s_free}.

\begin{figure}[h!]
\centering
\subfigure[The proposed method]{\includegraphics[width=0.8\linewidth]{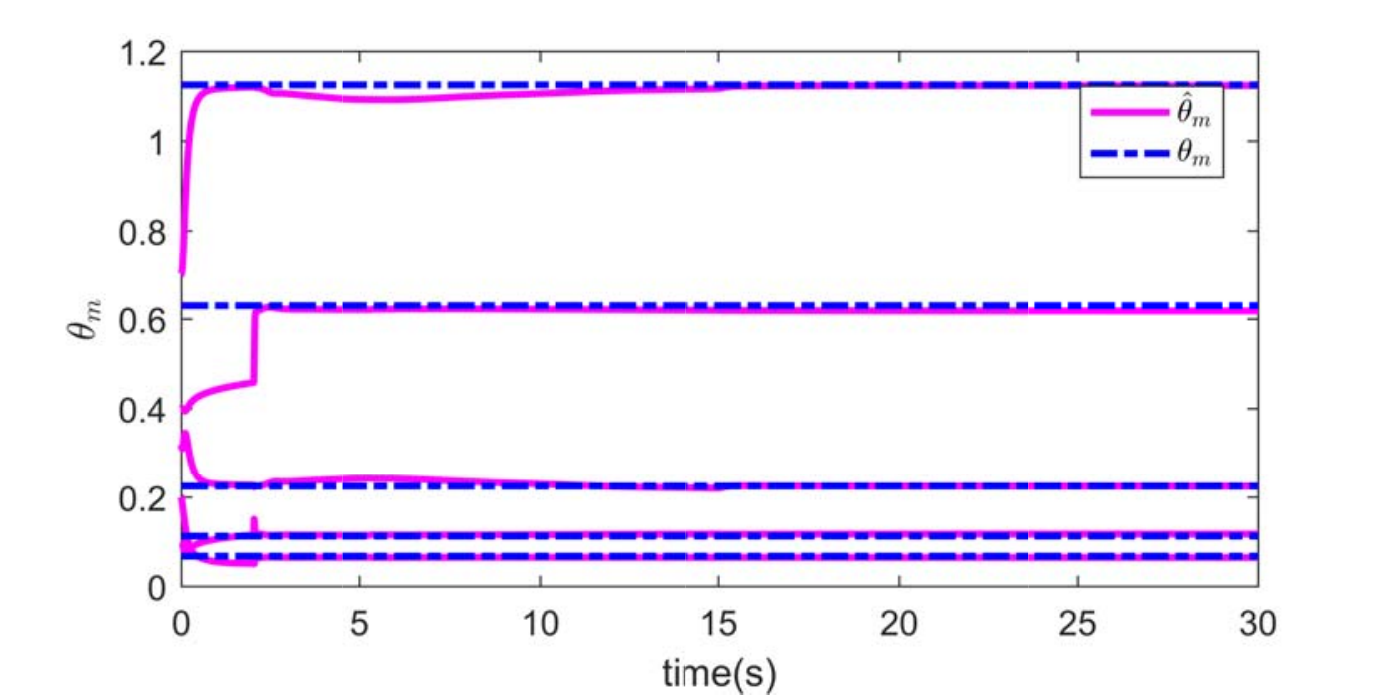}\label{fig:theta_m_contact_CAC}}
\subfigure[The method in \cite{SARRAS2014_JFI}]{\includegraphics[width=0.8\linewidth]{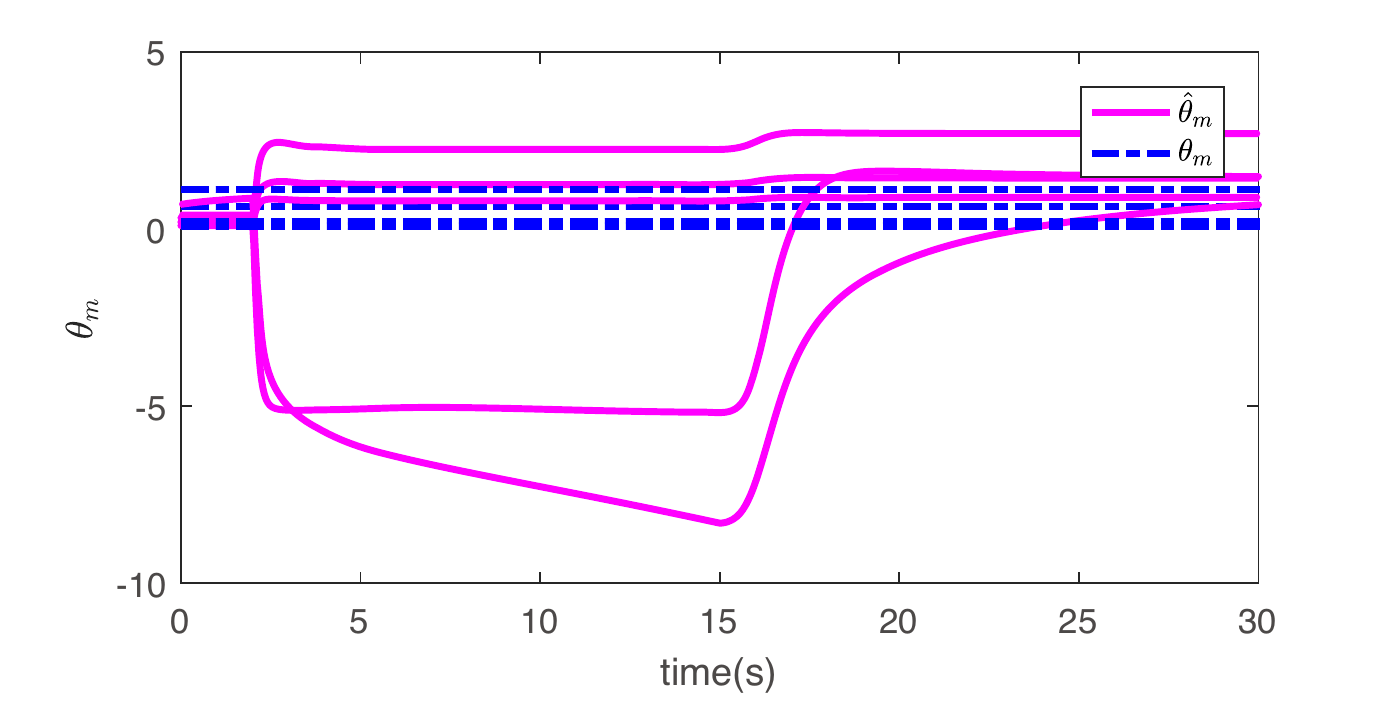}\label{fig:theta_m_contact_17}}
\caption{Scenario B: The dynamic parameter $\theta_m$ and the parameter estimate $\hat{\theta}_m$}
\label{fig:theta_m_contact}
\end{figure}

\begin{figure}[h!]
\centering
\subfigure[The proposed method]{\includegraphics[width=0.8\linewidth]{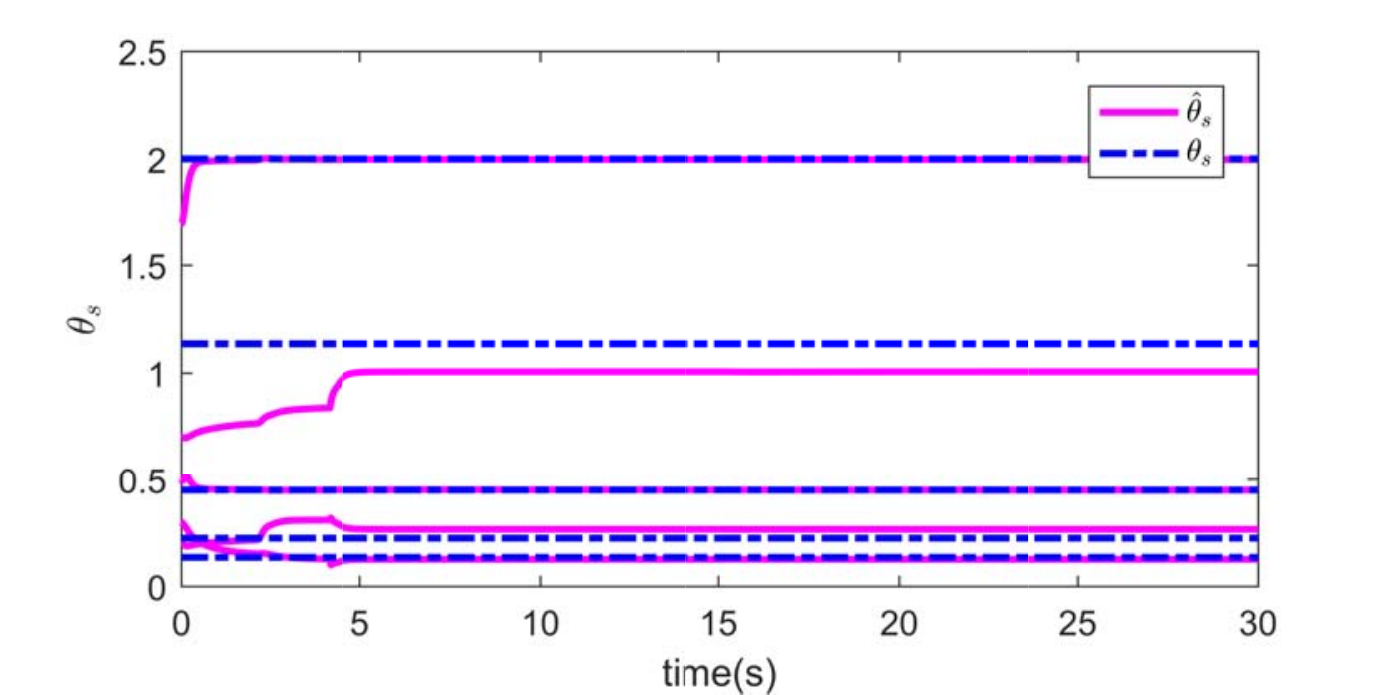}\label{fig:theta_s_contact_CAC}}
\subfigure[The method in \cite{SARRAS2014_JFI}]{\includegraphics[width=0.8\linewidth]{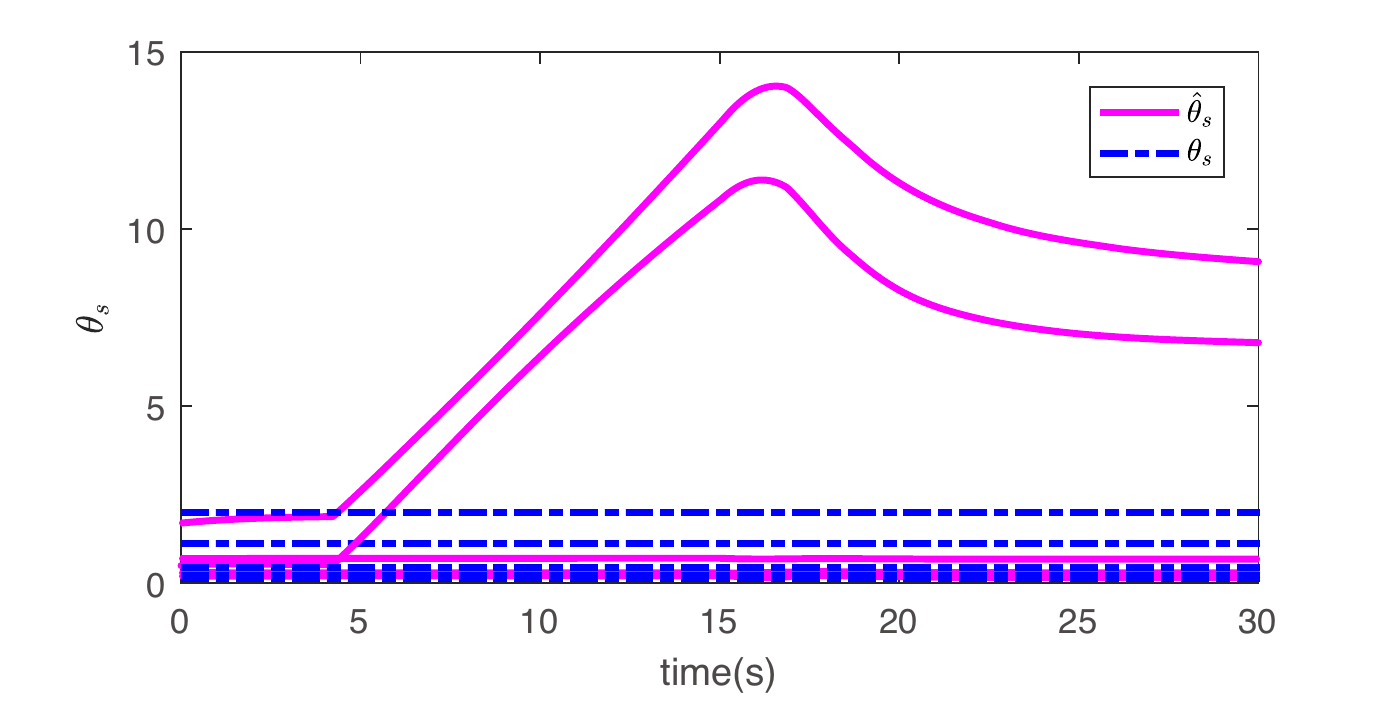}\label{fig:theta_s_contact_17}}
\caption{Scenario B: The dynamic parameter $\theta_s$ and the parameter estimate $\hat{\theta}_s$}
\label{fig:theta_s_contact}
\end{figure}

When the simulation was implemented in Scenario B, unfortunately, the controller proposed in \cite{SARRAS2014_JFI}  with the human input depicted in Figure \ref{fig:fhy_free} cannot drive the slave's end effector moving to the position $y=0.6$ m, as it is shown in Figure \ref{fig:noparaconverge_pos_y_free}, and thus the contact with the wall at $y=0.6$m could never happen in this circumstance. This means that the proposed method in \cite{SARRAS2014_JFI} needs larger force to drive the teleoperation system moving. To have a better comparison between the method in \cite{SARRAS2014_JFI} and the proposed method in this paper, we have to enlarge the human force for the proposed method in \cite{SARRAS2014_JFI}, which is given by $\bar{F}_h=10F_h$, where $\bar{F}_h$ is new human force exerted for the teleoperation system with the controller given in \cite{SARRAS2014_JFI}, while the other conditions keep the same.

The dynamic parameter estimates in Scenario B under the proposed method in this paper with the human input $F_h$ and the method in \cite{SARRAS2014_JFI} with the human input $\bar{F}_h$ are depicted in Figure~\ref{fig:theta_m_contact} and Figure~\ref{fig:theta_s_contact}.
By Figure ~\ref{fig:theta_m_free} and Figure ~\ref{fig:theta_m_contact}, the master's parameter estimates under the proposed algorithm in this paper converge to the true values very quickly in both Scenarios.   The slave's dynamic estimates also have a better tracking performance with its true values under our proposed method from Figure ~\ref{fig:theta_s_free} and Figure ~\ref{fig:theta_s_contact}.


To better describe the tracking performance, we use the proposed tracking performance measures  to evaluate the position tracking performance and the force tracking performance.  Table \ref{tab:Delta_J} shows the tracking performance of the teleoperation system under the proposed method in this paper and in \cite{SARRAS2014_JFI} by using our proposed performance measures. It is obviously that when the simulation is completed in Scenario A, that is, the slave is in free motion, and there is no environmental force, the position tracking performance under our proposed method is better than the one under the controller proposed in \cite{SARRAS2014_JFI}. When the simulation is implemented in Scenario B, both position tracking performance and force tracking performance could be measured by using the proposed performance measures. It is easy to find that in this case our method still has advantage over the prosed method in \cite{SARRAS2014_JFI}.
\begin{table}[h!]\centering
\caption{$\Delta_J$}
\begin{tabular}{ccc@{}ccc}
\hline
\multirow{2}{*}{$\Delta_p(\Delta_f)$} &  \multicolumn{2}{c}{The proposed method}&&\multicolumn{2}{c}{the method in \cite{SARRAS2014_JFI}}\\
\cline{2-3}\cline{5-6}
& Scenario A& Scenario B && Scenario A & Scenario B\\
\hline
$\Delta_p^1$& 3.0677& 7.5566&&  32.9066& 17.6223\\
\hline
$\Delta_p^2$& 7.8993& 80.3351&&  43.7317& 31.7531\\
\hline
$\Delta_f^1$& -& 16.3386&&  -& 20.9547\\
\hline
$\Delta_f^2$& -& 20.9843&&  -& 25.8604\\
\hline
\end{tabular}\label{tab:Delta_J}
\end{table}

%

%
%

\section{Conclusion}\label{sec:conclusion}
In this paper, a novel composite adaptive control framework that addressed dynamic uncertainties and time-varying delays for nonlinear teleoperation systems was  proposed.  A new defined prediction error was used to guarantee the parameter convergence. The stability criteria in terms of LMIs, which give the sufficient conditions related to the controller gains and  the upper bound of time delays, have been provided. To better describe the tracking performance, new tracking measures are proposed. The controller performance is verified via simulations. Further studies for parameter-converging adaptive control of teleoperation systems with both dynamic and kinematic uncertainties is under study and the results will be reported in the near future.


%


\section*{Acknowledgments}
This work was jointly supported by
the National Natural Science Foundation of China (No.~61333002,~61773053),  the Fundamental Research Funds for the Central Universities of USTB  (No.~FRF-TP-16-024A1,~FPR-BD-16-005A,~FRF-GF-17-A4),  the Beijing Key Discipline Development Program (No.~XK100080537), and the Beijing Natural Science Foundation (No.~4182039).

 \section*{References}
\bibliographystyle{elsarticle-num}
\bibliography{Sampling}

\label{lastpage}
\end{document}